\documentclass[graybox]{svmult}

\usepackage{type1cm}        % activate if the above 3 fonts are
\usepackage{makeidx}         % allows index generation
\usepackage{graphicx}        % standard LaTeX graphics tool
\usepackage{multicol}         % used for the two-column index
\usepackage[bottom]{footmisc}% places footnotes at page bottom
\usepackage{newtxtext}       % 
\usepackage{newtxmath}     % selects Times Roman as basic font
\usepackage[numbers]{natbib}
\usepackage{amsmath} % assumes amsmath package installed
\usepackage{amssymb}  % assumes amsmath package installed
\usepackage[percent]{overpic}
\usepackage{epstopdf}
\usepackage{tikz}
\usepackage{pict2e}
\usetikzlibrary{calc}
\usetikzlibrary{automata,positioning}
\usetikzlibrary{shapes.misc}
\definecolor{myorange}{rgb}{.984,.392,.016}
\definecolor{myblue}{rgb}{.106,.2,.322}

\tikzset{
cross/.style={cross out, draw=black, minimum size=2*(#1-\pgflinewidth), inner sep=0pt, outer sep=0pt},
cross/.default={1pt}}

\makeindex     

\newcommand{\R}{\mathbb{R}}
\newcommand{\Z}{\mathbb{Z}}
\newcommand{\Gaussian}{\mathbb{N}}
\newcommand{\SymMat}{\mathbb{S}}
\newcommand{\Exp}{\mathbb{E}}
\newcommand{\trace}{\mathrm{Tr}}
\newcommand{\diag}{\mathrm{diag}}
\newcommand{\cov}{\mathrm{cov}}
\newcommand{\nn}{\nonumber}
\DeclareMathOperator*{\argmin}{argmin\;}

\newcommand{\playerR}{\texttt{R}}
\newcommand{\playerS}{\texttt{S}}

\newcommand{\bU}{\bar{U}}
\newcommand{\bLambda}{\bar{\Lambda}}
\newcommand{\blambda}{\bar{\lambda}}

\newcommand{\tQ}{\tilde{Q}}

\newcommand{\game}{{\cal G}}

\newcommand{\calL}{{\cal L}}

\newcommand{\calX}{{\cal X}}
\newcommand{\calY}{{\cal Y}}
\newcommand{\calS}{{\cal S}}
\newcommand{\calU}{{\cal U}}

\newcommand{\ra}{\pmb{a}}
\newcommand{\rb}{\pmb{b}}
\newcommand{\rn}{\pmb{n}}

\newcommand{\rx}{\pmb{x}}
\newcommand{\rs}{\pmb{s}}
\newcommand{\rt}{\pmb{t}}
\newcommand{\ru}{\pmb{u}}
\newcommand{\rv}{\pmb{v}}
\newcommand{\ry}{\pmb{y}}
\newcommand{\rz}{\pmb{z}}
\newcommand{\rw}{\pmb{w}}

\newcommand{\vZeros}{\mathbf{0}}
\newcommand{\teta}{\tilde{\eta}}

\begin{document}

\title*{Deception-As-Defense Framework for Cyber-Physical Systems}
\author{Muhammed O. Sayin and Tamer Ba\c{s}ar}
\institute{Muhammed O. Sayin and Tamer Ba\c{s}ar \at Coordinated Science Laboratory, University of Illinois at Urbana-Champaign, 1308 West Main St., Urbana, IL, 61801, USA, \email{sayin2@illinois.edu},\email{basar1@illinois.edu}}

\maketitle

\abstract{We introduce deceptive signaling framework as a new defense measure against advanced adversaries in cyber-physical systems. In general, adversaries look for system-related information, e.g., the underlying state of the system, in order to learn the system dynamics and to receive useful feedback regarding the success/failure of their actions so as to carry out their malicious task. To this end, we craft the information that is accessible to adversaries strategically in order to control their actions in a way that will benefit the system, indirectly and without any explicit enforcement. Under the solution concept of game-theoretic hierarchical equilibrium, we arrive at a semi-definite programming problem equivalent to the infinite-dimensional optimization problem faced by the defender while selecting the best strategy when the information of interest is Gaussian and both sides have quadratic cost functions. The equivalence result holds also for the scenarios where the defender can have partial or noisy measurements or the objective of the adversary is not known. We show the optimality of linear signaling rule within the general class of measurable policies in communication scenarios and also compute the optimal linear signaling rule in control scenarios.}

\section{Introduction}
\label{sec:introduction}

\begin{quote}
All warfare is based on deception. Hence, when we are able to attack, we must seem unable; when using our forces, we must appear inactive; when we are near, we must make the enemy believe we are far away; when far away, we must make him believe we are near. \\
\null\hfill- Sun Tzu, The Art of War \cite{ref:Sunzi03}
\end{quote}

As quoted above, even the earliest known work on military strategy and war, The Art of War, emphasizes the importance of deception in security. Deception can be used as a defense strategy by making the opponent/adversary to perceive certain information of interest in an engineered way. Indeed, deception is also not limited to hostile environments. In all non-cooperative multi-agent environments, as long as there is asymmetry of information and one agent is informed about the information of interest while the other is not, then the informed agent has power on the uninformed one to manipulate his/her decisions or perceptions by sharing that information strategically. 

Especially with the introduction of cyber connectedness in physical systems, certain communication and control systems can be viewed as multi-agent environments, where each agent makes rational decisions to fulfill certain objectives. As an example, we can view transmitters (or sensors) and receivers (or controllers) as individual agents in communication (or control) systems. However, classical communication and control theory is based on the cooperation between these agents to meet certain challenges together, such as in mitigating the impact of a noisy channel in communication or in stabilizing the underlying state of a system around an equilibrium through feedback in control. However, cyber connectedness makes these multi-agent environments vulnerable against adversarial interventions and there is an inherent asymmetry of information as the information flows from transmitters (or sensors) to receivers (or controllers)\footnote{In control systems, we can also view the control input as information that flows implicitly from the controllers to the sensors since it impacts the underlying state and correspondingly the sensors' measurements.}. Therefore, if these agents are not cooperating, e.g., due to adversarial intervention, then the informed agents, i.e., transmitters or sensors, could seek to deceive the uninformed ones, i.e., receivers or controllers, so that they would perceive the underlying information of interest in a way the deceiver has desired, and correspondingly would take the manipulated actions.

Our goal, here, is to craft the information that could be available to an adversary in order to control his/her perception about the underlying state of the system as a defensive measure. The malicious objective and the normal operation of the system may not be completely opposite of each other as in the framework of a zero-sum game, which implies that there is a part of malicious objective that is benign and the adversary would be acting in line with the system's interest with respect to that aligned part of the objectives. If we can somehow restrain the adversarial actions to fulfill only the aligned part, then the adversarial actions, i.e., the attack, could inadvertently end up helping the system toward its goal. Since a rational adversary would make decisions based on the information available to him, the strategic crafting of the signal that is shared with the adversary, or the adversary can have access to, can be effective in that respect. Therefore, our goal is to design the information flowing from the informed agents, e.g., sensors, to the uninformed ones, e.g., controllers, in view of the possibility of adversarial intervention, so as to control the perception of the adversaries about the underlying system, and correspondingly to {\em persuade} them (without any explicit enforcement) to fulfill the aligned parts of the objectives as much as possible without fulfilling the misaligned parts. 

In this chapter, we provide an overview of the recent results  \cite{ref:Sayin17b,ref:Sayin17a,ref:Sayin19a} addressing certain aspects of this challenge in non-cooperative communication and control settings. For a discrete-time Gauss Markov process, and when the sender and the receiver in a non-cooperative communication setting have misaligned quadratic objectives, in \cite{ref:Sayin17b}, we have shown the optimality of linear signaling rules\footnote{We use the terms ``strategy", ``signaling/decision rule", and ``policy" interchangeably.} within the general class of measurable policies and provided an algorithm to compute the optimal policies numerically. Also in \cite{ref:Sayin17b}, we have formulated the optimal linear signaling rule in a non-cooperative linear-quadratic-Gaussian (LQG) control setting when the sensor and the controller have known misaligned control objectives. In \cite{ref:Sayin17a}, we have introduced a secure sensor design framework, where we have addressed the optimal linear signaling rule again in a non-cooperative LQG setting when the sensor and private-type controller have misaligned control objectives in a Bayesian setting, i.e., the distribution over the private type of the controller is known. In \cite{ref:Sayin19a}, we have addressed the optimal linear robust signaling in a non-Bayesian setting, where the distribution over the private type of the controller is not known, and provided a comprehensive formulation by considering also the cases where the sensor could have partial or noisy information on the signal of interest and relevance. We elaborate further on these results in some detail throughout the chapter.

In Section \ref{sec:review}, we review the related literature in economics and engineering. In Sections \ref{sec:framework} and \ref{sec:game}, we introduce the framework and formulate the deception-as-defense game, respectively. In Section \ref{sec:Gaussian}, we elaborate on Gaussian information of interest in detail. In Sections \ref{sec:communication} and \ref{sec:control}, we address the optimal signaling rules in non-cooperative communication and control systems. In Section \ref{sec:uncertainty}, we provide the optimal signaling rule against the worst possible distribution over the private types of the uninformed agent. In Section \ref{sec:noisy}, we extend the results to partial or noisy measurements of the underlying information of interest. Finally, we conclude the chapter in Section \ref{sec:conclusion} with several remarks and possible research directions.

{\em Notation:} Random variables are denoted by bold lower case letters, e.g., $\rx$. For a random vector $\rx$, $\cov\{\rx\}$ denotes the corresponding covariance matrix. For an ordered set of parameters, e.g., $x_1,\ldots,x_{\kappa}$, we use the notation $x_{k:l} = x_l,\ldots,x_k$, where $1\leq l \leq k \leq \kappa$. $\Gaussian(0,\cdot)$ denotes the multivariate Gaussian distribution with zero mean and designated covariance. For a vector $x$ and a matrix $A$, $x'$ and $A'$ denote their transposes, and $\|x\|$ denotes the Euclidean $\ell_2$-norm of the vector $x$. For a matrix $A$, $\trace\{A\}$ denotes its trace. We denote the identity and zero matrices with the associated dimensions by $I$ and $O$, respectively. $\SymMat^m$ denotes the set of $m$-by-$m$ symmetric matrices. For positive semi-definite matrices $A$ and $B$, $A\succeq B$ means that $A-B$ is also positive semi-definite. 

\section{Deception Theory in Literature}\label{sec:review}

There are various definitions of deception. Depending on the specific definition at hand, the analysis or the related applications vary. Commonly in signaling-based deception definitions, there is an information of interest private to an informed agent whereas an uninformed agent may benefit from that information to make a certain decision. If the informed and uninformed agents are strategic while, respectively, sharing information and making a decision, then the interaction can turn into a game where the agents select their strategies according to their own objectives while taking into account the fact that the other agent would also have selected his/her strategy according to his/her different objective. Correspondingly, such an interaction between the informed and uninformed agents can be analyzed under a game-theoretic solution concept. Note that there is a main distinction between incentive compatible deception model and deception model with policy commitment.

\begin{definition}
We say that a deception model is {\bf incentive compatible} if neither the informed nor the uninformed agent have an incentive to deviate from their strategies unilaterally. 
\end{definition}

The associated solution concept here is Nash equilibrium \cite{ref:Basar99}. Existence of a Nash equilibrium is not guaranteed in general. Furthermore, even if it exists, there may also be multiple Nash equilibria. Without certain commitments, any of the equilibria may not be realized or if one has been realized, which of them would be realized is not certain beforehand since different ones could be favorable for different players.  

\begin{definition}
We say that in a deception model, there is {\bf policy commitment} if either the informed or the uninformed agent commits to play a certain strategy beforehand and the other agent reacts being aware of the committed strategy. 
\end{definition}

The associated solution concept is Stackelberg equilibrium, where one of the players leads the game by announcing his/her committed strategy \cite{ref:Basar99}. Existence of a Stackelberg equilibrium is not guaranteed in general over unbounded strategy spaces. However, if it exists, all the equilibria would lead to the same game outcome for the leader of the game since the leader could have always selected the favorable one among them. We also note that if there is a favorable outcome for the leader in the incentive compatible model, the leader has the freedom to commit to that policy in the latter model. Correspondingly, the leader is advantageous by acting first to commit to play according to a certain strategy even though the result may not be incentive compatible. 

Game theoretical analysis of deception has attracted substantial interest in various disciplines, including economics and engineering fields. In the following subsections, we review the literature in these disciplines with respect to models involving incentive compatibility and policy commitment.   

\subsection{Economics Literature}

The scheme of the type introduced above, called strategic information transmission, was introduced in a seminal paper by V. Crawford and J. Sobel in \cite{ref:Crawford82}. This has attracted significant attention in the economics literature due to the wide range of relevant applications, from advertising to expert advise sharing. In the model adopted in \cite{ref:Crawford82}, the informed agent's objective function includes a commonly known bias term different from the uninformed agent's objective. That bias term can be viewed as the misalignment factor in-between the two objectives. For the incentive compatible model, the authors have shown that all equilibria are partition equilibria, where the informed agent controls the resolution of the information shared via certain quantization schemes, under certain assumptions on the objective functions (satisfied by quadratic objectives), and the assumption that the information of interest is drawn from a bounded support.

Following this inaugural introduction of the strategic information transmission framework, also called {\em cheap talk} due to the costless communication over an ideal channel, different settings, such as
\begin{itemize}
\item Single sender and multiple receivers \cite{ref:Farrell86,ref:Farrell96},
\item Multiple senders and single receiver \cite{ref:Gilligan89,ref:Krishna00},
\item Repeated games \cite{ref:Morris01}, 
\end{itemize}
have been studied extensively; however, all have considered the scenarios where the underlying information is one-dimensional, e.g., a real number. However, multi-dimensional information can lead to interesting results like full revelation of the information even when the misalignment between the objectives is arbitrarily large if there are multiple senders with different bias terms, i.e., misalignment factors \cite{ref:Battaglini02}. Furthermore, if there is only one sender yet multidimensional information, there can be full revelation of information at certain dimensions while at the other dimensions, the sender signals partially in a partition equilibrium depending on the misalignment between the objectives \cite{ref:Battaglini02}. 

The ensuing studies \cite{ref:Farrell86,ref:Gilligan89,ref:Farrell96,ref:Krishna00,ref:Morris01,ref:Battaglini02} on cheap talk \cite{ref:Crawford82} have analyzed the incentive compatibility of the players. More recently, in \cite{ref:Kamenica11}, the authors have proposed to use a deception model with policy commitment. They call it ``sender-preferred sub-game perfect equilibrium" since the sender cannot distort or conceal information once the signal realization is known, which can be viewed as the sender revealing and committing to the signaling rule in addition to the corresponding signal realization. For information of interest drawn from a compact metric space, the authors have provided necessary and sufficient conditions for the existence of a strategic signal that can benefit the informed agent, and characterized the corresponding optimal signaling rule. Furthermore, in \cite{ref:Tamura14}, the author has shown the optimality of linear signaling rules for multivariate Gaussian information of interest and with quadratic objective functions.

\subsection{Engineering Literature}

There exist various engineering applications depending on the definition of deception. Reference \cite{ref:Pawlick17} provides a taxonomy of these studies with a specific focus on security. Obfuscation techniques to hide valuable information, e.g., via externally introduced noise \cite{ref:Howe09,ref:Clark12, ref:Zhu12} can also be viewed as deception based defense. As an example, in \cite{ref:Howe09}, the authors have provided a browser extension that can obfuscate user's real queries by including automatically-fabricated queries to preserve privacy. Here, however, we specifically focus on signaling-based deception applications, in which we craft the information available to adversaries to control their perception rather than corrupting it. In line with the browser extension example, our goal is to persuade the query trackers to perceive the user behavior in a certain fabricated way rather than limiting their ability to learn the actual user behavior.   

In computer security, various (heuristic) deception techniques, e.g., honeypots and honey nets, are prevalent to make the adversary perceive a honey-system as the real one or a real system as a honey-one \cite{ref:Spitzner02}. Several studies, e.g., \cite{ref:Carroll11}, have analyzed honeypots within the framework of binary signaling games by abstracting the complexity of crafting a real system to be perceived as a honeypot (or crafting a honeypot to be perceived as a real system) to binary signals. However, here, our goal is to address the optimal way to craft the underlying information of interest with a continuum support, e.g., a Gaussian state. 

The recent study \cite{ref:Saritas17} addresses strategic information transmission of multivariate Gaussian information over an additive Gaussian noise channel for quadratic misaligned cost functions and identifies the conditions where the signaling rule attaining a Nash equilibrium can be a linear function. Recall that for scalar case, when there is no noisy channel in-between, all the equilibria are partition equilibria, implying all the signaling rules attaining a Nash equilibrium are nonlinear except babbling equilibrium, where the informed agent discloses no information \cite{ref:Crawford82}. Two other recent studies \cite{ref:Akyol17} and \cite{ref:Farokhi17} address strategic information transmission for the scenarios where the bias term is not common knowledge of the players and the solution concept is Stackelberg equilibrium rather than Nash equilibrium. They have shown that the Stackelberg equilibrium could be attained by linear signaling rules under certain conditions, different from the partition equilibria in the incentive compatible cheap talk model \cite{ref:Crawford82}. In \cite{ref:Farokhi17}, the authors have studied strategic sensor networks for multivariate Gaussian information of interest and with myopic quadratic objective functions in dynamic environments and by restricting the receiver's strategies to affine functions. In \cite{ref:Akyol17}, for jointly Gaussian scalar private information and bias variable, the authors have shown that optimal sender strategies are linear functions within the general class of measurable policies for misaligned quadratic cost functions when there is an additive Gaussian noise channel and hard power constraint on the signal, i.e., when it is no longer cheap talk.

\section{Deception-As-Defense Framework}
\label{sec:framework}

\begin{figure}[t!]
\includegraphics[width=\textwidth]{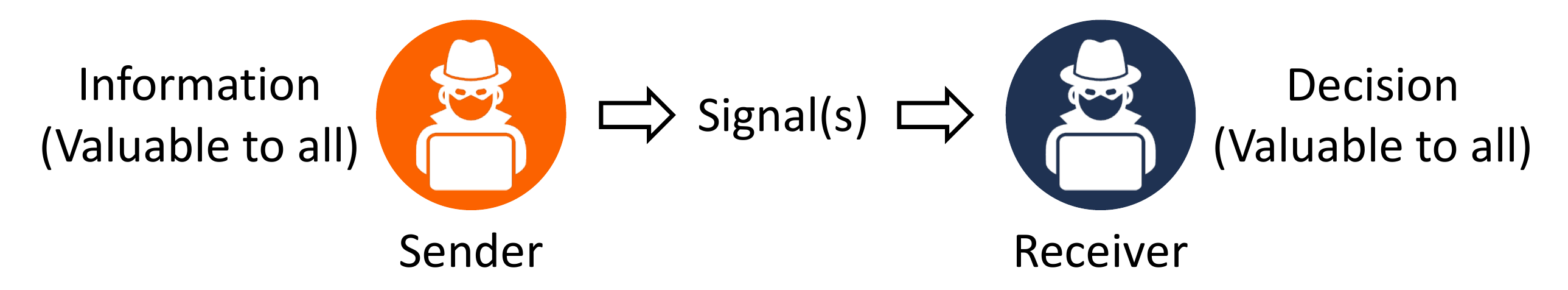}
\caption{Strategic information disclosure.} 
\label{fig:model}
\end{figure}

Consider a multi-agent environment with asymmetry of information, where each agent is a selfish decision maker taking action or actions to fulfill his/her own objective only while actions of any agent could impact the objectives of the others. As an example, Fig. \ref{fig:model} illustrates a scenario with two agents: Sender (\playerS) and Receiver (\playerR), where \playerS~has access to (possibly partial or noisy version of) certain information valuable to \playerR, and \playerS~sends a signal or signals related to the information of interest to \playerR. 

\begin{definition}
We say that an informed agent (or the signal the agent crafts) is {\bf deceptive} if he/she shapes the information of interest private to him/her strategically in order to control the perception of the uninformed agent by {\em removing}, {\em changing}, or {\em adding contents}. 
\end{definition}

Deceptive signaling can play a key role in multi-agent non-cooperative environments as well as in cooperative ones, where certain (uninformed) agents could have been compromised by certain adversaries. In such scenarios, informed agents can signal strategically to the uninformed ones in case they could have been compromised. Furthermore, deceiving an adversary to act, or attack the system in a way aligned with the system's goals can be viewed as being too optimistic due to the very definition of adversary. However, an adversary can also be viewed as a selfish decision maker seeking to satisfy a certain malicious objective, which may not necessarily be completely conflicting with the system's objective. This now leads to the following notion of ``deception-as-defense". 

\begin{definition}
We say that an informed agent engages in a {\bf deception-as-defense} mode of operation if he/she crafts the information of interest strategically to persuade the uninformed malicious agent (without any explicit enforcement) to act in line with the aligned part of the objective as much as possible without taking into account the misaligned part.  
\end{definition}

We re-emphasize that this approach differs from the approaches that seek to raise suspicion on the information of interest to sabotage the adversaries' malicious objectives. Sabotaging the adversaries' malicious objectives may not necessarily be the best option for the informed agent unless the objectives are completely opposite of each other. In this latter case, the deception-as-defense framework actually ends up seeking to sabotage the adversaries' malicious objectives. 

We also note that this approach differs from lying, i.e., the scenario where the informed agent provides a totally different information (correlated or not) {\em as if} it is the information of interest. Lying could be effective, as expected, as long as the uninformed agent trusts the legitimacy of the provided information. However, in non-cooperative environments, this could turn into a game where the uninformed agent becomes aware of the possibility of lying. This correspondingly raises suspicion on the legitimacy of the shared information and could end up sabotaging the adversaries' malicious objectives rather than controlling their perception of the information of interest. 

Once a defense mechanism has been widely deployed, this can cause the advanced adversaries learn the defense policy in the course of time. Correspondingly, the solution concept of policy commitment model can address this possibility in the deception-as-defense framework in a robust way if the defender commits to a certain policy that takes into account the best reaction of the adversaries that are aware of the policy. Furthermore, the transparency of the signal sent via the committed policy generates a trust-based relationship in-between \playerS~and \playerR, which is powerful to persuade \playerR~to make certain decisions inadvertently without any explicit enforcement by \playerS. 

\section{Game Formulation}\label{sec:game}

The information of interest is considered to be a realization of a known, continuous, random variable in static settings or a known (discrete-time) random process in dynamic settings. Since the static setting is a special case of the dynamic setting, we formulate the game in a dynamic, i.e., multi-stage, environment. We denote the information of interest by $\{\rx_k\in\calX\}$, where $\calX\subset \R^m$ denotes its support. Let $\{\rx_k\}$ have zero mean and (finite) second-order moment $\Sigma_k := \cov\{\rx_k\}\in\SymMat^{m}$. We consider the scenarios where each agent has perfect recall and constructs his/her strategy accordingly. \playerS~has access to a possibly partial or noisy version of the information of interest, $\rx_k$. We denote the noisy measurement of $\rx_k$ by $\ry_k\in\calY$, where $\calY\subset\R^m$ denotes its support. For each instance of the information of interest, \playerS~selects his/her signal as a second-order random variable 
\begin{equation}
\rs_k = \eta_k(\ry_{1:k}),
\end{equation}
correlated with $\ry_{1:k}$, but not necessarily determined through a deterministic transformation on $\ry_{1:k}$ (i.e., $\eta_k(\cdot)$ is in general a random mapping). Let us denote the set of all signaling rules by $\Upsilon_k$. As we will show later, when we allow for such randomness in the signaling rule, under certain conditions the solution turns out to be a linear function of the underlying information $\ry_{1:k}$ and an additive independent noise term. Due to the policy commitment by \playerS, at each instant, with perfect recall, \playerR~selects a Borel measurable decision rule $\gamma_k:\calS^k\rightarrow\calU$, where $\calU\subset\R^{r}$, from a certain policy space $\Gamma_k$ in order to make a decision\begin{equation}
\ru_k = \gamma_k(\rs_{1:k}),
\end{equation}  
knowing the signaling rules $\{\eta_k\}$ and observing the signals sent $\rs_{1:k}$.

Let $\kappa$ denote the length of the horizon. We consider that the agents have cost functions to minimize, instead of utility functions to maximize. Clearly, the framework could also be formulated accordingly for utility maximization rather straightforwardly. Furthermore, we specifically consider that the agents have quadratic cost functions, denoted by $U_{\playerS}(\eta_{1:\kappa},\gamma_{1:\kappa})$ and $U_{\playerR}(\eta_{1:\kappa},\gamma_{1:\kappa})$. 

\begin{example}{An Example in Non-cooperative Communication Systems}
Over a finite horizon with length $\kappa$, \playerS~seeks to minimize over $\eta_{1:\kappa}\in\Upsilon := \bigtimes_{k=1}^{\kappa}\Upsilon_{k}$
\begin{eqnarray}
U_{\playerS}(\eta_{1:\kappa},\gamma_{1:\kappa}) &= &\Exp\left\{\sum_{k=1}^{\kappa} \|Q_{\playerS}\rx_k - R_{\playerS}\gamma_k(\eta_1(\ry_1),\ldots,\eta_k(\ry_{1:k}))\|^2\right\}\nn\\
&= &\Exp\left\{\sum_{k=1}^{\kappa} \|Q_{\playerS}\rx_k - R_{\playerS}\ru_k\|^2\right\},\label{eq:SobjComm}
\end{eqnarray}
by taking into account that \playerR~seeks to minimize over $\gamma_{1:\kappa}\in\Gamma:=\bigtimes_{k=1}^{\kappa}\Gamma_{k}$
\begin{eqnarray}\label{eq:RobjComm}
U_{\playerR}(\eta_{1:\kappa},\gamma_{1:\kappa}) &= &\Exp\left\{\sum_{k=1}^{\kappa} \|Q_{\playerR}\rx_k - R_{\playerR}\gamma_k(\eta_1(\ry_1),\ldots,\eta_k(\ry_{1:k}))\|^2\right\}\nn\\
&= &\Exp\left\{\sum_{k=1}^{\kappa} \|Q_{\playerR}\rx_k - R_{\playerR}\ru_k\|^2\right\},
\end{eqnarray}
where the weight matrices are arbitrary (but fixed). The following special case illustrates the applicability of this general structure of misaligned objectives \eqref{eq:SobjComm} and \eqref{eq:RobjComm}. Suppose that the information of interest consists of two separate processes $\{\rz_k\}$ and $\{\rt_k\}$, e.g., $\rx_k := \begin{bmatrix} \rz_k' & \rt_k' \end{bmatrix}'$. Then \eqref{eq:SobjComm} and \eqref{eq:RobjComm} cover the scenarios where 
\playerR~seeks to estimate $\rz_k$ by minimizing
\begin{equation}\label{eq:RobjComm2}
\Exp\left\{\sum_{k=1}^{\kappa} \|\rz_k - \ru_k\|^2\right\},
\end{equation}
whereas \playerS~wants \playerR~to perceive $\rz_k$ as $\rt_k$, and end up minimizing
\begin{equation}\label{eq:SobjComm2}
\Exp\left\{\sum_{k=1}^{\kappa} \|\rt_k - \ru_k\|^2\right\}.
\end{equation}
\end{example}

\begin{example}{An Example in Non-cooperative Control Systems}\label{example:control}
Consider a controlled Markov process, e.g., 
\begin{equation}\label{eq:stateControlled}
\rx_{k+1} = A\rx_{k} + B\ru_k +\rw_k,
\end{equation}
where $\rw_k\sim\Gaussian(0,\Sigma_w)$ is a white Gaussian noise process. \playerS~seeks to minimize over $\eta_{1:\kappa}\in\Upsilon$
\begin{equation}\label{eq:SobjCont}
U_{\playerS}(\eta_{1:\kappa},\gamma_{1:\kappa}) = \Exp\left\{\sum_{k=1}^{\kappa} \rx_{k+1}'Q_{\playerS}^{}\rx_{k+1}^{} + \ru_k' R_{\playerS}^{}\ru_k^{}\right\},
\end{equation}
by taking into account that \playerR~seeks to minimize over $\gamma_{1:\kappa}\in\Gamma$
\begin{equation}\label{eq:RobjCont}
U_{\playerR}(\eta_{1:\kappa},\gamma_{1:\kappa}) = \Exp\left\{\sum_{k=1}^{\kappa} \rx_{k+1}'Q_{\playerR}^{}\rx_{k+1}^{} + \ru_k' R_{\playerR}^{}\ru_k^{}\right\},
\end{equation}
with arbitrary (but fixed) positive semi-definite matrices $Q_{\playerS}$ and $Q_{\playerR}$, and positive-definite matrices $R_{\playerS}$ and $R_{\playerR}$. Similar to the example in communication systems, this general structure of misaligned objectives \eqref{eq:SobjCont} and \eqref{eq:RobjCont} can bring in interesting applications. Suppose the information of interest consists of two separate processes $\{\rz_k\}$ and $\{\rt_k\}$, e.g., $\rx_k := \begin{bmatrix} \rz_k' & \rt_k' \end{bmatrix}'$, where $\{\rt_k\}$ is an exogenous process, which does not depend on \playerR's decision $\ru_k$. For certain weight matrices, \eqref{eq:SobjCont} and \eqref{eq:RobjCont} cover the scenarios where \playerR~seeks to regularize $\{\rz_k\}$ around zero vector by minimizing 
\begin{equation}
\Exp\left\{\sum_{k=1}^{\kappa} \rz_{k+1}'\rz_{k+1}^{} + \ru_k'\ru_k^{}\right\},
\end{equation}
whereas \playerS~seeks \playerR~to regularize $\{\rz_{k}\}$ around the exogenous process $\{\rt_k\}$ by minimizing
\begin{equation}
\Exp\left\{\sum_{k=1}^{\kappa} (\rz_{k+1}-\rt_{k+1})'(\rz_{k+1}-\rt_{k+1}) + \ru_k'\ru_k^{}\right\}.
\end{equation}
\end{example}

We define the deception-as-defense game as follows:

\begin{definition}
The {\bf deception-as-defense game} $\game := (\Upsilon,\Gamma,\{\rx_k\},\{\ry_k\},U_{\playerS},U_{\playerR})$ is a Stackelberg game between \playerS~and \playerR, where
\begin{itemize}
\item $\{\rx_k\}$ denotes the information of interest,
\item $\{\ry_k\}$ denotes \playerS's (possibly noisy) measurements of the information of interest,
\item $U_{\playerS}$ and $U_{\playerR}$ are the objective functions of \playerS~and \playerR, defined respectively by \eqref{eq:SobjComm} and \eqref{eq:RobjComm}, or \eqref{eq:SobjCont} and \eqref{eq:RobjCont}.
\end{itemize}
Under the deception model with policy commitment, \playerS~is the leader, who announces (and commits to) his strategies beforehand, while \playerR~is the follower, reacting to the leader's announced strategies. Since \playerR~is the follower and takes actions knowing \playerS's strategy $\eta_{1:\kappa}\in\Upsilon$, we let $B(\eta_{1:\kappa}) \subset \Gamma$ be \playerR's best reaction set to \playerS's strategy $\eta_{1:\kappa}\in\Upsilon$. Then, the strategy and best reaction pair $(\eta_{1:\kappa}^*,B(\eta_{1:\kappa}^*))$ attains the Stackelberg equilibrium provided that
\begin{eqnarray}
&&\eta_{1:\kappa}^* \in\argmin_{\eta_{1:\kappa}\in \Upsilon}\max_{\gamma_{1:\kappa}\in B(\eta_{1:\kappa})} U_{\playerS}(\eta_{1:\kappa},\gamma_{1:\kappa}),\\
&& B(\eta_{1:\kappa}) \;= \argmin_{\gamma_{1:\kappa}\in \Gamma} U_{\playerR}(\eta_{1:\kappa},\gamma_{1:\kappa}).
\end{eqnarray}
\end{definition}

\section{Quadratic Costs and Information of Interest}
\label{sec:Gaussian}

Misaligned {\em quadratic} cost functions, in addition to their various applications, play an essential role in the analysis of the game $\game$. One advantage is that a quadratic cost function can be written as a linear function of the covariance of the posterior estimate of the underlying information of interest. Furthermore, when the information of interest is Gaussian, we can formulate a necessary and sufficient condition on the {\em covariance of the posterior estimate}, which turns out to be just semi-definite matrix inequalities. This leads to an equivalent semi-definite programming (SDP) problem over a {\em finite} dimensional space instead of finding the best signaling rule over an infinite-dimensional policy space. In the following, we elaborate on these observations in further detail. 

Due to the policy commitment, \playerS~needs to anticipate \playerR's reaction to the selected signaling rule $\eta_{1:\kappa}\in\Upsilon$. Here, we will focus on the non-cooperative communication system, and later in Section \ref{sec:control}, we will show how we can transform a non-cooperative control setting into a non-cooperative communication setting under certain conditions. Since the information flow is in only one direction, \playerR~faces the least mean square error problem for given $\eta_{1:\kappa}\in\Upsilon$. Suppose that $R_{\playerR}'R_{\playerR}^{}$ is invertible. Then, the best reaction by \playerR~is given by
\begin{equation}
\gamma_k^*(\rs_{1:k}) = (R_{\playerR}'R_{\playerR}^{})^{-1}R_{\playerR}'Q_{\playerR}^{} \Exp\{\rx_k | \rs_{1:k}\},
\end{equation} 
almost everywhere over $\R^{r}$. Note that the best reaction set $B(\eta_{1:\kappa})$ is a singleton and the best reaction is linear in the posterior estimate $\Exp\{\rx_k | \rs_{1:k}\}$, i.e., the conditional expectation of $\rx_k$ with respect to the random variables $\rs_{1:k}$. When we substitute the best reaction by \playerR~into \playerS's cost function, we obtain
\begin{equation}\label{eq:subsCost}
\sum_{k=1}^{\kappa}\Exp\|Q_{\playerS}\rx_k - M_{\playerS}\Exp\{\rx_k|\rs_{1:\kappa}\}\|^2,
\end{equation}
where $M_{\playerS}^{}:= R_{\playerS}^{}(R_{\playerR}'R_{\playerR}^{})^{-1}R_{\playerR}'Q_{\playerR}^{}$. Since for arbitrary random variables $\ra$ and $\rb$, 
\begin{equation}\label{eq:basic}
\Exp\{\ra\Exp\{\ra|\rb\}\} = \Exp\{\Exp\{\ra|\rb\}\Exp\{\ra|\rb\}\},
\end{equation}
the objective function to be minimized by \playerS, \eqref{eq:subsCost}, can be written as
\begin{equation}\label{eq:short}
\sum_{k=1}^{\kappa}\Exp\|Q_{\playerS}\rx_k - M_{\playerS}\Exp\{\rx_k|\rs_{1:\kappa}\}\|^2 = \sum_{k=1}^{\kappa} \trace\{H_k V\} + c,
\end{equation}
where $H_k := \cov\{\Exp\{\rx_k|\rs_{1:k}\}\}$ denotes the covariance of the posterior estimate, 
\begin{equation}
V := M_{\playerS}'M_{\playerS}^{} - M_{\playerS}'Q_{\playerS}^{} - Q_{\playerS}'M_{\playerS}^{}
\end{equation}
and the constant $c$ is given by
\begin{equation}
c := \sum_{k=1}^{\kappa}\trace\{Q_{\playerS}'Q_{\playerS}^{} \Sigma_k^{}\}.
\end{equation}
We emphasize that $H_k\in\SymMat^m$ is not the posterior covariance, i.e.,  $\cov\{\Exp\{\rx_k|\rs_{1:k}\}\}\neq\cov\{\rx_k|\rs_{1:k}\}$ in general.

The cost function depends on the signaling rule $\eta_{1:\kappa}\in\Upsilon$ only through the covariance matrices $H_{1:\kappa}$ and the cost is an affine function of $H_{1:\kappa}$. By formulating the relation, we can obtain an equivalent finite-dimensional optimization problem over the space of symmetric matrices as an alternative to the infinite-dimensional problem over the policy space $\Upsilon$. Next, we seek to address the following question.

\begin{question}{Relation between $\eta_{1:\kappa}$ and $H_{1:\kappa}$}
What is the relation between the signaling rule $\eta_{1:\kappa}\in\Upsilon$ and the covariance of the posterior estimate $H_{1:\kappa}$? 
\end{question}

Here, we only consider the scenario where \playerS~has access to the underlying information of interest perfectly. We will address the scenarios with partial or noisy measurements in Section \ref{sec:noisy} by transforming that setting to the setting of perfect measurements. 

There are two extreme cases for the shared information: either sharing the information fully without any crafting or sharing no information. The former one implies that the covariance of the posterior estimate would be $\Sigma_k$ whereas the latter one implies that it would be $\cov\{\Exp\{\rx_k|\rs_{1:k-1}\}\}$ since \playerR~has perfect memory. 

\begin{question}{In-between the extremes of sharing everything and sharing nothing}
What would $H_k\in\SymMat^{m}$ be if \playerS~has shared the information only partially?
\end{question}

To address this, if we consider the positive semi-definite matrix $\cov\{\rx_k - \Exp\{\rx_k|\rs_{1:k}\}\}$, by \eqref{eq:basic} we obtain
\begin{equation}\label{eq:upp}
\cov\{\rx_k - \Exp\{\rx_k|\rs_{1:k}\}\} = \Sigma_k - H_k.
\end{equation}
Furthermore, if we consider the positive semi-definite matrix $\cov\{\Exp\{\rx_k|\rs_{1:k}\} - \Exp\{\rx_k|\rs_{1:k-1}\}\}$, by \eqref{eq:basic} we obtain
\begin{equation}\label{eq:bel}
\cov\{\Exp\{\rx_k|\rs_{1:k}\} - \Exp\{\rx_k|\rs_{1:k-1}\}\} = H_k - \cov\{\Exp\{\rx_k|\rs_{1:k-1}\}\}.
\end{equation}
Therefore, based on \eqref{eq:upp} and \eqref{eq:bel}, we obtain the necessary condition:
\begin{equation}\label{eq:necessary}
\Sigma_k \succeq H_k \succeq \cov\{\Exp\{\rx_k|\rs_{1:k-1}\}\},
\end{equation}
which is independent of the distribution of the underlying information and the policy space of \playerS. 

\begin{question}{Sufficient Condition}
What would be the sufficient condition? Is the necessary condition on $H_k\in\SymMat^m$ \eqref{eq:necessary} sufficient?
\end{question}

The sufficient condition for arbitrary distributions is an open problem. However, in the following subsection, we show that when information of interest is Gaussian, we can address the challenge and the necessary condition turns out to be sufficient. 

\subsection{Gaussian Information of Interest}

In addition to its use in modeling various uncertain phenomena based on the central limit theorem, Gaussian distribution has special characteristics which make it versatile in various engineering applications, e.g., in communication and control. The deception-as-defense framework is not an exception for the versatility of the Gaussian distribution. As an example, if the information of interest is Gaussian, the optimal signaling rule turns out to be a linear function within the general class of measurable policies, as to be shown in different settings throughout this chapter.

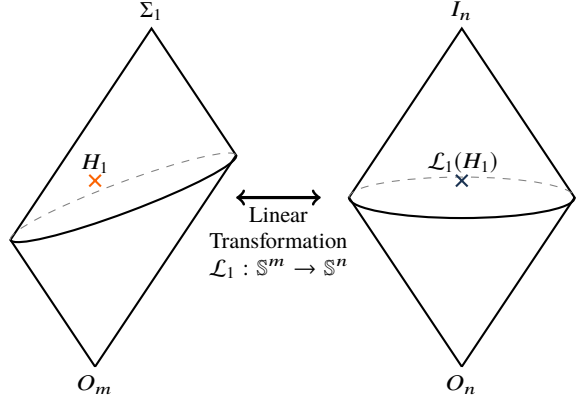
\begin{figure}[t!]
\sidecaption[t]
\begin{tikzpicture}
	\def\scale{.75}
	\def\shift{6.5*\scale}
	\def\w{2*\scale}
	\def\l{3*\scale}
	\def\h{sqrt((2*\w)^2 + (.5*\l)^2)/2}
	\def\angle{atan((.5*\l)/(2*\w))}
	\draw[thick] (0,0) node[below]{$O_m$} -- ({\w*1.25},{\l*1.25}) -- ({\w*.5},{2*\l}) node[above]{$\Sigma_1$} -- ({-\w*.75},{\l*.75}) -- (0,0);
	\draw[thick,rotate around={{\angle}:({-\w*.75},{\l*.75})}] ({-\w*.75},{\l*.75}) arc (-180:0:{\h} and {\h/8});% bottom half of the ellipse
	\draw[dashed,color=gray,rotate around={{\angle}:({-\w*.75},{\l*.75})}] ({-\w*.75},{\l*.75}) arc (180:0:{\h} and {\h/8});% right half of the left ellipse
	\draw[thick] ({\shift},0) node[below]{$O_n$} -- ({\w+\shift},{\l}) -- ({\shift},{2*\l}) node[above]{$I_n$} -- ({-\w+\shift},{\l}) -- ({\shift},0);
	\draw[thick] ({-\w+\shift},{\l}) arc (-180:0:{\w} and {\w/5.5});% bottom half of the ellipse
	\draw[dashed,color=gray] ({-\w+\shift},{\l}) arc (180:0:{\w} and {\w/5.5});% right half of the left ellipse
	\draw ({0},{\l*1.1}) node[cross=3pt,thick,myorange] {} node[above]{$H_1$};
	\draw ({\shift},{\l*1.1}) node[cross=3pt,thick,myblue] {} node[above]{$\calL_1(H_1)$};
	\draw [<->,very thick] ( {1.25*\w},{\l}) -- ({2*\w},{\l}) node[text width=2cm,text centered, midway,below]{Linear\\Transformation\\ $\calL_1:\SymMat^m\rightarrow\SymMat^n$};
\end{tikzpicture}
\caption{A figurative illustration that the covariance of the posterior estimate $H_1$ is bounded from above and below by the semi-cones in the space of symmetric matrices, i.e., $\Sigma_1\succeq H_1 \succeq O_m$. Furthermore, we can transform the space to the form at the right figure through certain linear mapping $\calL_1:\SymMat^m\rightarrow\SymMat^n$, where $n\in\Z$ may be different from $m\in\Z$.}\label{fig:neces}
\end{figure}

Let us first focus on the single-stage setting, where the necessary condition \eqref{eq:necessary} is given as
\begin{equation}\label{eq:neces1}
\Sigma_1 \succeq H_1 \succeq O_m. 
\end{equation}
The convention here is that for arbitrary symmetric matrices $A,B\in\SymMat^m$, $A\succeq B$ means that $A-B\succeq O$, that is positive semi-definite. We further note that the space of positive-semi-definite matrices is a semi-cone \cite{ref:Wolkowicz00}. Correspondingly, Fig. \ref{fig:neces} provides a figurative illustration of \eqref{eq:neces1}, where $H_1\in\SymMat^m$ is bounded from both below and above by certain semi-cones in the space of symmetric matrices.

With a certain linear transformation bijective over \eqref{eq:neces1}, denoted by $\calL_1 : \SymMat^m \rightarrow \SymMat^n$, where $n\in\Z$ is not necessarily the same with $m\in\Z$, the necessary condition \eqref{eq:neces1} can be written as
\begin{equation}
I_n \succeq \calL_1(H_1) \succeq O_n.
\end{equation}
As an example of such a linear mapping when $\Sigma_1\in\SymMat^m$ is invertible, we can consider $\calL_1(H_1) = \Sigma_1^{-1/2}H_1\Sigma_1^{-1/2}$ and $n=m$. If $\Sigma_1$ is singular, then the following lemma from \cite{ref:Sayin18e} plays an important role to compute such a linear mapping. 

\begin{lemma}\label{lem:semidef}
Provided that a given positive semi-definite matrix can be partitioned into blocks such that a block at the diagonal is a zero matrix, then certain off-diagonal blocks must also be zero matrices, i.e.,
\begin{equation}
\begin{bmatrix} A&B\\B'& O \end{bmatrix} \succeq O \Leftrightarrow A\succeq O \mbox{ and } B = O.
\end{equation}
\end{lemma}

Let the singular $\Sigma_1\in\SymMat^m$ with rank $n<m$ have the eigen-decomposition
\begin{equation}\label{eq:eig}
\Sigma_1 = U_1 \begin{bmatrix} \Lambda_1 & O \\ O & O \end{bmatrix}U_1',
\end{equation}
where $\Lambda_1\succ O_n$. Then, \eqref{eq:neces1} can be written as
\begin{eqnarray}\label{eq:MM}
\begin{bmatrix} \Lambda_1 & O \\ O & O \end{bmatrix} - \begin{bmatrix} N_{1,1} & N_{1,2} \\ N_{1,2}' & N_{2,2} \end{bmatrix} = \begin{bmatrix} \Lambda_1 - N_{1,1} & -N_{1,2} \\ -N_{1,2}' & -N_{2,2}\end{bmatrix} \succeq O,
\end{eqnarray}
where we let
\begin{equation}
U_1'H_1U_1 = \begin{bmatrix} N_{1,1} & N_{1,2} \\ N_{1,2}' & N_{2,2} \end{bmatrix}
\end{equation}
be the corresponding partitioning, i.e., $N_{1,1}\in\SymMat^n$. Since $U_1'H_1U_1\succeq O_m$, the diagonal block $N_{2,2}\in\SymMat^{m-n}$ must be positive semi-definite \cite{ref:Horn85}. Further, \eqref{eq:MM} yields that $-N_{2,2} \succeq O_{m-n}$, which implies that $N_{2,2} = O_{m-n}$. Invoking Lemma \ref{lem:semidef}, we obtain $N_{1,2} = O_{n\times(m-n)}$. Therefore, a linear mapping bijective over \eqref{eq:neces1} is given by
\begin{equation}\label{eq:map}
\calL_1(H_1) = \begin{bmatrix}\Lambda_1^{-1/2} & O_{n\times(m-n)} \end{bmatrix} U_1'H_1U_1 \begin{bmatrix}\Lambda_1^{-1/2} \\ O_{(m-n)\times n} \end{bmatrix},
\end{equation}
where the unitary matrix $U_1 \in\R^{m\times m}$ and the diagonal matrix $\Lambda_1\in\SymMat^n$ are as defined in \eqref{eq:eig}.

\begin{svgraybox}
With the linear mapping \eqref{eq:map} that is bijective over \eqref{eq:neces1}, the necessary condition on $H_1\in\SymMat^m$ can be written as
\begin{eqnarray}
\Sigma_1\succeq H_1 \succeq O_m &\Leftrightarrow& I_n \succeq \calL_1(H_1) \succeq O_n\nn\\
&\Rightarrow& \mbox{Eigenvalues of $\calL_1(H_1)$ are in the closed interval $[0,1]$}\nn
\end{eqnarray}
since the eigenvalues of $I_n$ weakly majorize the eigenvalues of the positive semi-definite $\calL_1(H_1)$ from below \cite{ref:Horn85}. 
\end{svgraybox}

Up to this point, the specific distribution of the information of interest did not play any role. However, for the sufficiency of the condition \eqref{eq:neces1}, Gaussianness of the information of interest plays a crucial role as shown in the following theorem \cite{ref:Sayin19a}.

\begin{theorem}\label{theorem:equivalent}
Consider $m$-variate Gaussian information of interest $\rx_1 \sim \Gaussian (0,\Sigma_1)$. Given any stochastic kernel $\eta_1\in\Upsilon_1$, we have
\begin{equation}
\Sigma_1 \succeq \cov\{\Exp\{\rx_1|\eta_1(\rx_1)\}\} \succeq O_m.
\end{equation}
Furthermore, given any covariance matrix $H_1\in\SymMat^m$ satisfying
\begin{equation}\label{eq:nec}
\Sigma_1 \succeq H_1 \succeq O_m,
\end{equation}
we have that there exists a probabilistic {\bf linear}-in-$\rx_1$ signaling rule
\begin{equation}\label{eq:signal}
\eta_1^{}(\rx_1^{}) = L_1'\rx_1^{} + \rn_1^{},
\end{equation}
where $L_1\in\R^{m\times m}$ and $\rn_1^{} \sim \Gaussian (0,\Sigma_1^o)$ is an independent $m$-variate Gaussian random variable, such that $\cov\{\Exp\{\rx_1|\eta_1(\rx_1)\}\} = H_1$. Let $\calL_1(H_1) \in \SymMat^n$ have the eigen-decomposition $\calL_1^{}(H_1^{}) = \bU_1^{} \bLambda_1^{} \bU_1'$ and $\bLambda_1 = \diag\{\blambda_{1,1},\ldots,\blambda_{1,n}\}$. Then, the corresponding matrix $L_1\in\R^{m\times m}$ and the covariance $\Sigma_1^o \succeq O_m$ are given by
\begin{equation}\label{eq:L1}
L_1 := U_1 \begin{bmatrix} I_n \\ O \end{bmatrix}  \Lambda_1^{-1/2}\bU_1\Lambda_1^o \begin{bmatrix} I_n & O \end{bmatrix},
\end{equation}
where the unitary matrix $U_1\in\R^{m\times m}$ and the diagonal matrix $\Lambda_1\in\SymMat^n$ are as defined in \eqref{eq:eig}, $\Lambda_1^o := \diag\{\lambda_{1,1}^o,\ldots, \lambda_{1,n}^o\}$, $\Sigma_1^o = \diag\{(\sigma_{1,1}^o)^2, \ldots, (\sigma_{1,n}^o)^2, 0, \ldots,0\}$, and
\begin{equation}\label{eq:lambda1}
\frac{(\lambda_{1,i}^o)^2}{(\lambda_{1,i}^o)^2 + (\sigma_{1,i}^o)^2} = \blambda_{1,i} \in[0,1],\;\forall\;i=1,\ldots,n.
\end{equation}
\end{theorem}

\begin{proof}
Note that for Gaussian information and the signaling rule \eqref{eq:signal}, the covariance of the posterior estimate is given by
\begin{equation}
\cov\{\Exp\{\rx_1^{}|L_1'\rx_1^{} + \rn_1^{}\}\} = \Sigma_1^{}L_1^{}(L_1'\Sigma_1^{}L_1^{} + \Sigma_1^o)^{\dagger}L_1'\Sigma_1^{}
\end{equation}
Given $H_1\in\SymMat^m$ satisfying \eqref{eq:nec}, for \eqref{eq:L1} and \eqref{eq:lambda1}, the linear-in-$\rx_1$ signaling rule \eqref{eq:signal} yields that $\cov\{\Exp\{\rx_1^{}|L_1'\rx_1^{} + \rn_1^{}\}\}= H_1^{}$.
\end{proof}

\begin{important}{Implication of Theorem \ref{theorem:equivalent}}
If the underlying information of interest is Gaussian, instead of the functional optimization problem
\begin{equation}\label{eq:functional}
\min_{\eta_1\in\Upsilon} \Exp\|Q_{\playerS}\rx_1 - K_{\playerS}\Exp\{\rx_1|\eta_1(\rx_1)\}\|^2,
\end{equation}
we can consider the equivalent finite-dimensional problem
\begin{equation}\label{eq:SDP}
\min_{S\in\SymMat} \trace\{SV\}, \mbox{ subject to } \Sigma_1\succeq S \succeq O.
\end{equation}
Then, we can compute the optimal signaling rule $\eta_1^*$ corresponding to the solution of \eqref{eq:SDP} via \eqref{eq:signal}-\eqref{eq:lambda1}.
\end{important}

\begin{svgraybox}
Without any need to solve the functional optimization problem \eqref{eq:functional}, Theorem \ref{theorem:equivalent} shows the optimality of the ``linear plus a random variable" signaling rule within the general class of stochastic kernels when the information of interest is Gaussian. 
\end{svgraybox}

\begin{warning}{Versatility of the Equivalence}
Furthermore, a linear signaling rule would still be optimal even when we introduce additional constraints on the covariance of the posterior since the equivalence between \eqref{eq:functional} and \eqref{eq:SDP} is not limited with the equivalence in optimality.
\end{warning}

Recall that the distribution of the underlying information plays a role only in proving the sufficiency of the necessary condition. Therefore, in general, based on only the necessary condition, we have
\begin{equation}
\min_{\eta_1\in\Upsilon} \Exp\|Q_{\playerS}\rx_1 - K_{\playerS}\Exp\{\rx_1|\eta_1(\rx_1)\}\|^2 \geq \min_{S\in\SymMat} \trace\{SV\}, \mbox{ subject to } \Sigma_1\succeq S \succeq O. 
\end{equation}
The equality holds when the information of interest is Gaussian. 

\begin{svgraybox}
Therefore, for fixed covariance $\Sigma_1\in\SymMat^m$, Gaussian distribution is the best one for \playerS~to persuade \playerR~in accordance with his/her deceptive objective, since it yields total {\em freedom} to attain any covariance of the posterior estimate in-between the two extremes $\Sigma_1\succeq H_1 \succeq O$.
\end{svgraybox}

The following counter example shows that the sufficiency of the necessary condition \eqref{eq:nec} holds only in the case of the Gaussian distribution.

\begin{example}{A Counter Example for Arbitrary Distributions}
For a clear demonstration, suppose that $m=2$ and $\Sigma_1 = I_2$, and correspondingly $\rx_1 = \begin{bmatrix} \rx_{1,1} & \rx_{1,2}\end{bmatrix}'$. The covariance matrix $H:=\begin{bmatrix} 1 & 0 \\ 0 & 0 \end{bmatrix}$ satisfies the necessary condition \eqref{eq:nec} since
\begin{equation}
I_2 \succeq H \succeq O_2,
\end{equation}
which implies that the signal $\rs_1$ must be fully informative about $\rx_{1,1}$ without giving any information about $\rx_{1,2}$. Note that $\Sigma_1=I_2$ only implies that $\rx_{1,1}$ and $\rx_{1,2}$ are uncorrelated, yet not necessarily independent for arbitrary distributions. Therefore, if $\rx_{1,1}$ and $\rx_{1,2}$ are uncorrelated but dependent, then any signaling rule cannot attain that covariance of the posterior estimate even though it satisfies the necessary condition.
\end{example}

Let us now consider a Gauss-Markov process, which follows the following first-order auto-regressive recursion
\begin{equation}\label{eq:state}
\rx_{k+1} = A\rx_k + \rw_k,
\end{equation}
where $A\in\R^{m\times m}$ and $\rw_k\sim\Gaussian(0,\Sigma_w)$. For this model, the necessary condition \eqref{eq:necessary} is given by
\begin{equation}\label{eq:necesk}
\Sigma_k \succeq H_k \succeq AH_{k-1}A',
\end{equation}
for $k=2,\ldots,\kappa$. Given $H_{1:k-1}$, let $\Sigma_k - AH_{k-1}A'$ have the eigen-decomposition
\begin{equation}\label{eq:eigk}
\Sigma_k - AH_{k-1}A' = U_k \begin{bmatrix} \Lambda_k & O \\ O & O \end{bmatrix} U_k',
\end{equation}
where $\Lambda_k\succ O_{n_k}$, i.e., $\Sigma_k - AH_{k-1}A'$ has rank $n_k$. The linear transformation $\calL_k:\bigtimes_{i=1}^k\SymMat^m\rightarrow \SymMat^n$ given by
\begin{equation}\label{eq:Lk}
\calL_k(H_{1:k}) = \begin{bmatrix}\Lambda_k^{1/2} & O_{m-n_k} \end{bmatrix} U_k'(H_k - AH_{k-1}A')U_k \begin{bmatrix}\Lambda_k^{1/2} \\ O_{m-n_k} \end{bmatrix}
\end{equation}
is bijective over \eqref{eq:necesk}. With the linear mapping \eqref{eq:Lk}, the necessary condition on $H_{1:\kappa}\in\bigtimes_{i=1}^{\kappa}\SymMat^m$ can be written as
\begin{equation}
\Sigma_k \succeq H_k \succeq AH_{k-1}A' \Leftrightarrow I_{n_k} \succeq \calL_k(H_{1:k}) \succeq O_{n_k},
\end{equation}
which correspondingly yields that $\calL_k(H_{1:k})\in\SymMat^{n_k}$ has eigenvalues in the closed interval $[0,1]$. Then, the following theorem extends the equivalence result of the single-stage to multi-stage ones \cite{ref:Sayin19a}.

\begin{theorem}\label{theorem:multi}
Consider the $m$-variate Gauss-Markov process $\{\rx_k \sim \Gaussian (0,\Sigma_k)\}$ following the state recursion \eqref{eq:state}. Given any stochastic kernel $\eta_k\in\Upsilon_k$ for $k=1,\ldots,\kappa$, we have
\begin{eqnarray}
\Sigma_1 &\succeq& \cov\{\Exp\{\rx_1|\rs_1\}\} \succeq O_m\\
\Sigma_k &\succeq& \cov\{\Exp\{\rx_k|\rs_{1:k}\}\} \succeq A \cov\{\Exp\{\rx_{k-1}|\rs_{1:k-1}\}\} A', \;k=2,\ldots,\kappa.
\end{eqnarray}
Furthermore, given any covariance matrices $H_{1:\kappa}\in\bigtimes_{i=1}^{\kappa}\SymMat^m$ satisfying
\begin{equation}\label{eq:nec}
\Sigma_k \succeq H_k \succeq AH_{k-1}A',
\end{equation}
where $H_0=O_m$, then there exists a probabilistic {\bf linear}-in-$\rx_k$, i.e., {\bf memoryless}, signaling rule 
\begin{equation}
\eta_k^{}(\rx_{1:k}^{}) = L_k'\rx_k^{} + \rn_k^{},
\end{equation}
where $L_k\in\R^{m\times m}$ and $\{\rn_k^{} \sim \Gaussian (0,\Sigma_k^o)\}$ is independently distributed $m$-variate Gaussian process such that $\cov\{\Exp\{\rx_k|\eta_1(\rx_1),\ldots,\eta_{k}(\rx_{1:k})\}\} = H_k$ for all $k=1,\ldots,\kappa$. Given $H_{1:k-1}$, let $\calL_k(H_{1:k}) \in \SymMat^{n_k}$ have the eigen-decomposition $\calL_k^{}(H_{1:k}^{}) = \bU_k^{} \bLambda_k^{} \bU_k'$ and $\bLambda_k = \diag\{\blambda_{k,1},\ldots,\blambda_{k,n_k}\}$. Then, the corresponding matrix $L_k\in\R^{m\times m}$ and the covariance $\Sigma_k^o \succeq O_m$ are given by
\begin{equation}\label{eq:LLk}
L_k := U_k \begin{bmatrix} I_{n_k} \\ O \end{bmatrix}  \Lambda_k^{-1/2}\bU_k^{}\Lambda_k^o \begin{bmatrix} I_{n_k} & O \end{bmatrix},
\end{equation}
where the unitary matrix $U_k\in\R^{m\times m}$ and the diagonal matrix $\Lambda_k\in\SymMat^{n_k}$ are defined in \eqref{eq:eigk}, $\Lambda_k^o := \diag\{\lambda_{k,1}^o,\ldots, \lambda_{k,n_k}^o\}$, $\Sigma_k^o = \diag\{(\sigma_{k,1}^o)^2, \ldots, (\sigma_{k,n_k}^o)^2, 0, \ldots,0\}$, and
\begin{equation}\label{eq:lambdak}
\frac{(\lambda_{k,i}^o)^2}{(\lambda_{k,i}^o)^2 + (\sigma_{k,i}^o)^2} = \blambda_{k,i} \in[0,1],\;\forall\;i=1,\ldots,n_k.
\end{equation}
\end{theorem}

\begin{svgraybox}
Without any need to solve the functional optimization problem
\begin{equation}
\min_{\eta_{1:\kappa}\in\Upsilon} \sum_{k=1}^{\kappa} \Exp\|Q_{\playerS}\rx_k - K_{\playerS}\Exp\{\rx_k|\eta_1(\rx_1),\ldots,\eta_k(\rx_{1:k})\}\|^2,
\end{equation}
Theorem \ref{theorem:multi} shows the optimality of the ``linear plus a random variable" signaling rule within the general class of stochastic kernels also in dynamic environments, when the information of interest is Gaussian. 
\end{svgraybox}

\section{Communication Systems}
\label{sec:communication}

In this section, we elaborate further on the deception-as-defense framework in non-cooperative communication systems with a specific focus on Gaussian information of interest. We first note that in this case the optimal signaling rule turns out to be a linear deterministic signaling rule, where \playerS~does not need to introduce additional independent noise on the signal sent. Furthermore, the optimal signaling rule can be computed analytically for the single-stage game \cite{ref:Tamura14}. We also extend the result on the optimality of linear signaling rules to multi-stage ones \cite{ref:Sayin17b}.  

In the single stage setting, by Theorem \ref{theorem:equivalent}, the SDP problem equivalent to the problem \eqref{eq:subsCost} faced by \playerS~is given by
\begin{equation}\label{eq:equi}
\min_{S\in\SymMat^m} \trace\{SV\} \mbox{ subject to } \Sigma_1\succeq S \succeq O_m.
\end{equation}
We can have a closed form solution for the equivalent SDP problem \eqref{eq:subsCost} \cite{ref:Tamura14}. If $\Sigma_1\in\SymMat^m$ has rank $n$, then a change of variable with the linear mapping $\calL_1:\SymMat^m \rightarrow \SymMat^n$ \eqref{eq:map}, e.g., $T := \calL_1(S)$, yields that \eqref{eq:equi} can be written as
\begin{equation}\label{eq:equiCoV}
\min_{T\in\SymMat^n} \trace\{T W\} \mbox{ subject to } I_n \succeq T \succeq O_m,
\end{equation}
where
\begin{equation}
W := \begin{bmatrix} \Lambda_1^{1/2} & O_{n\times (m-n)}^{}\end{bmatrix}U_1' V U_1^{} \begin{bmatrix} \Lambda_1^{1/2} \\ O_{(m-n)\times n}^{} \end{bmatrix}.
\end{equation}
If we multiply each side of the inequalities in the constraint set of \eqref{eq:equiCoV} from left and right with unitary matrices such that the resulting matrices are still symmetric, the semi-definiteness inequality would still hold. Therefore, let the symmetric matrix $W\in\SymMat^n$ have the eigen-decomposition 
\begin{equation}\label{eq:Weig}
W = \begin{bmatrix}U_{+} & U_{-}\end{bmatrix}\begin{bmatrix} \Lambda_+ & O \\ O & -\Lambda_- \end{bmatrix} \begin{bmatrix} U_+' \\ U_-' \end{bmatrix},
\end{equation}
where $\Lambda_+$ and $\Lambda_-$ are positive semi-definite matrices with dimensions $n_+$ and $n_-$. Then \eqref{eq:equiCoV} could be written as
\begin{equation}\label{eq:blocks}
\min\limits_{\substack{T_+\in\SymMat^{n_+},\\ T_-\in\SymMat^{n_-}}} \trace\{T_+\Lambda_+\} - \trace\{T_-\Lambda_-\}\mbox{ subject to } I_{n_+} \succeq T_+\succeq O_{n_+}, I_{n_-} \succeq T_-\succeq O_{n_-}
\end{equation}
and there exists a $T_r\in\R^{n_+\times n_-}$ such that 
\begin{equation}\label{eq:T}
T = \begin{bmatrix} U_+ & U_- \end{bmatrix}\begin{bmatrix} T_{+}^{} & T_r^{} \\ T_r' & T_-^{} \end{bmatrix} \begin{bmatrix} U_+'\\ U_-' \end{bmatrix}
\end{equation}
satisfies the constraint in \eqref{eq:equiCoV}. Then, the following lemma shows that an optimal solution for \eqref{eq:blocks} is given by $T_+^* = O_{n_+}^{}$, $T_r^* = O_{n_+\times n_-}^{}$, and $T_-^* = O_{n_-}^{}$. Therefore, in \eqref{eq:blocks}, the second (negative semi-definite) term $-\trace\{T_{-}\Lambda_-\}$ can be viewed as the aligned part of the objectives whereas the remaining first (positive semi-definite) term $\trace\{T_+\Lambda_+\}$ is the misaligned part.

\begin{lemma}
For arbitrary $I_n\succeq A=[a_{i,j}] \succeq O_n$ and diagonal positive semi-definite $B=\diag\{b_1,\ldots,b_n\}\succeq O_n$, we have
\begin{equation}
0 \leq \trace\{AB\}=\sum_{i=1}^n a_{i,i}b_i \leq \trace\{B\} = \sum_{i=1}^n b_i.
\end{equation}
\end{lemma}

\begin{proof}
The left inequality follows since $\trace\{AB\} = \trace\{A^{1/2}BA^{1/2}\}$ while $A^{1/2}BA^{1/2}$ is positive semi-definite. The right inequality follows since the diagonal entries of $A$ are majorized from below by its eigenvalues by Schur Theorem \cite{ref:Horn85} while the eigenvalues of $A$ are weakly majorized from below by the eigenvalues of $I_n$ since $I_n\succeq A$ \cite{ref:Horn85}. 
\end{proof}

Based on \eqref{eq:T}, the solution for \eqref{eq:blocks} implies that the optimal solution for \eqref{eq:equiCoV} is given by 
\begin{equation}
T^* = \begin{bmatrix} U_+ & U_- \end{bmatrix}\begin{bmatrix} O_{n_+} & O_{n_+\times n_-} \\ O_{n_-\times n_+} & I_{n_-}\end{bmatrix}\begin{bmatrix} U_+'\\ U_-' \end{bmatrix}.
\end{equation}
By invoking Theorem \ref{theorem:equivalent} and \eqref{eq:L1}, we obtain the following theorem to compute the optimal signaling rule analytically in single-stage $\game$ (a version of the theorem can be found in \cite{ref:Tamura14}).

\begin{theorem}
Consider a single-stage deception-as-defense game $\game$, where \playerS~and \playerR~have the cost functions \eqref{eq:SobjComm} and \eqref{eq:RobjComm}, respectively. Then, an optimal signaling rule is given by
\begin{equation}\label{eq:optSignal}
\eta_1^*(\rx_1^{}) = \begin{bmatrix} I_n \\ O_{(m-n)\times n} \end{bmatrix} \begin{bmatrix} O_{n_+\times n} \\ U_-' \end{bmatrix} \Lambda_1^{-1/2} \begin{bmatrix} I_n & O_{n\times (m-n)} \end{bmatrix} U_1' \rx_1,
\end{equation}
almost everywhere over $\R^m$. The matrices $U_1\in\R^{m\times m}$, $\Lambda_1\in\SymMat^{n}$ are as defined in \eqref{eq:eig}, and $U_- \in \R^{n\times n_-}$ is as defined in \eqref{eq:Weig}.
\end{theorem}

Note that the optimal signaling rule \eqref{eq:optSignal} does not include any  additional noise term. The following corollary shows that the optimal signaling rule does not include additional noise when $\kappa>1$ as well (versions of this theorem can be found in \cite{ref:Sayin17b} and \cite{ref:Sayin18e}).

\begin{corollary}\label{corollary:multi}
Consider a deception-as-defense game $\game$, where the exogenous Gaussian information of interest follows the first-order autoregressive model \eqref{eq:state}, and the players \playerS~and \playerR~have the cost functions \eqref{eq:SobjComm} and \eqref{eq:RobjComm}, respectively. Then, for the optimal solution $S_{1:\kappa}^* \in \bigtimes_{k=1}^{\kappa}\SymMat^m$ of the equivalent problem, $P_k := \calL_k(S_{1:k}^*)$ is a symmetric idempotent matrix, which implies that the eigenvalues of $P_k\in\SymMat^{n_k}$ are either $0$ or $1$. Let $n_{k,1}\in\Z$ denote the rank of $P_k$, and $P_k$ have the eigen-decomposition 
\begin{equation}
P_k = \begin{bmatrix} U_{k,0} & U_{k,1} \end{bmatrix} \begin{bmatrix} O_{n_k-n_{k,1}} &  \\  & I_{n_{k,1}} \end{bmatrix} \begin{bmatrix} U_{k,0}' \\ U_{k,1}' \end{bmatrix}.
\end{equation}
Then, the optimal signaling rule is given by
\begin{equation}\label{eq:optSignalMulti}
\eta_k^*(\rx_{1:k}^{}) =\begin{bmatrix} I_{n_k} \\ O_{(m-n_k)\times n_k} \end{bmatrix}\begin{bmatrix} O_{(n_k-n_{k,1})\times n_k} \\ U_{k,1}'\end{bmatrix} \Lambda_k^{-1/2} \begin{bmatrix} I_{n_k} & O_{n_k \times (m-{n_k})}\end{bmatrix} U_k'\rx_k,
\end{equation}
almost everywhere over $\R^m$, for $k=1,\ldots,\kappa$. The unitary matrix $U_k\in\R^{m\times m}$ and the diagonal matrix $\Lambda_k\in\SymMat^{n_k}$ are defined in \eqref{eq:eigk}.
\end{corollary}

\section{Control Systems}
\label{sec:control}

The deception-as-defense framework also covers the non-cooperative control settings including a sensor observing the state of the system and a controller driving the system based on the sensor outputs according to certain quadratic control objectives, e.g., \eqref{eq:RobjCont}. Under the general game setting where the players can select any measurable policy, the control setting cannot be transformed into a communication setting straight-forwardly since the problem features non-classical information due to the asymmetry of information between the players and the dynamic interaction through closed-loop feedback signals, which leads to two-way information flow rather than one-way flow as in the communication setting in Section \ref{sec:communication}. However, the control setting can be transformed into a non-cooperative communication setting under certain conditions, e.g., when signaling rules are restricted to be linear plus a random term.

Consider a controlled Gauss-Markov process following the recursion \eqref{eq:stateControlled}, and with players \playerS~and \playerR~seeking to minimize the quadratic control objectives \eqref{eq:SobjCont} and \eqref{eq:RobjCont}, respectively. Then, by completing to squares, the cost functions \eqref{eq:SobjCont} and \eqref{eq:RobjCont} can be written as
\begin{equation}\label{eq:square}
\Exp\left\{\sum_{k=1}^{\kappa} \rx_{k+1}'Q_{j}^{}\rx_{k}^{} + \ru_k'R_{j}^{}\ru_k^{}\right\} = \sum_{k=1}^{\kappa}\Exp\|K_{j,k}\rx_k+\ru_k\|_{\Delta_{j,k}}^2 + \delta_{j,0},
\end{equation}
where $j=\playerS,\playerR$, and
\begin{eqnarray}
&K_{j,k}& = \Delta_{j,k}^{-1}B'\tQ_{j,k+1}A\\
&\Delta_{j,k}& = B'\tQ_{j,k+1} B + R_{j}\\
&\delta_{j,0}& = \trace\{Q_{j}\Sigma_1\} + \sum_{k=1}^{\kappa}\trace\{\tQ_{j,k+1}\Sigma_w\}
\end{eqnarray}
and $\{\tQ_{j,k}\}$ follows the discrete-time dynamic Riccati equation:
\begin{equation}
\tQ_{j,k} = Q_j + A'(\tQ_{j,k+1}-\tQ_{j,k+1}B\Delta_{j,k}^{-1}B'\tQ_{j,k+1})A,
\end{equation}
and $\tQ_{j,\kappa+1}= Q_{j}$.

On the right-hand side of \eqref{eq:square}, the state depends on the control input $\ru_{j,k}$, for $j=\playerS,\playerR$, however, a routine change of variables yields that
\begin{equation}\label{eq:change}
\sum_{k=1}^{\kappa}\Exp\|K_{j,k}\rx_k+\ru_k\|_{\Delta_{j,k}}^2 = \sum_{k=1}^{\kappa} \Exp\|K_{j,k}^{}\rx_k^o + \ru_k^o\|_{\Delta_{j,k}}^2,
\end{equation}
where we have introduced the control-free, i.e., exogenous, process $\{\rx_k^o\}$ following the first-order auto-regressive model
\begin{equation}\label{eq:free}
\rx_{k+1}^o = A\rx_k^o + \rw_k^{}, \; k=1,\ldots,\kappa,\mbox{ and }\rx_1^o = \rx_1^{},
\end{equation}
and a linearly transformed control input
\begin{equation}\label{eq:trans}
\ru_k^o = \ru_k^{} + K_{j,k}^{}B\ru_{k-1}^{} + \ldots + K_{j,k}^{}A^{k-2}B\ru_1^{}.
\end{equation}

\begin{important}{Non-classical Information Scheme under General Game Settings}
The right-hand side of \eqref{eq:change} resembles the cost functions in the communication setting, which may imply separability over the horizon and for
\begin{equation}
\Exp\|K_{j,k}^{}\rx_k^o + \ru_k^o\|_{\Delta_{j,k}}^2,
\end{equation}
the optimal transformed control input is given by $\ru_{k}^o = -K_{j,k}\Exp\{\rx_k^o| \rs_{1:k}^{}\}$ and the corresponding optimal control input could be computed by reversing the transformation \eqref{eq:trans}. However, here, the control rule constructs the control input based on the sensor outputs, which are chosen strategically by the non-cooperating \playerS~while \playerS~constructs the sensor outputs based on the actual state, which is driven by the control input, rather than the control-free state. Therefore, \playerR~can have impact on the sensor outputs by having an impact on the actual state. Therefore, the game $\game$ under the general setting features a non-classical information scheme. However, if \playerS's strategies are restricted to linear policies $\eta_k^{\ell}\in\Upsilon_k^{\ell}\subset\Upsilon_k$, given by
\begin{equation}
\eta^{\ell}_k(\rx_{1:k}) = L_{k,k}'\rx_k + \ldots + L_{k,1}'\rx_1 + \rn_k,
\end{equation}
then we have
\begin{eqnarray}
\Exp\{\rx_k^o | L_{k,k}'\rx_k + &\ldots& + L_{k,1}'\rx_1 + \rn_k,\ldots,L_{1,1}'\rx_1 + \rn_1\}\\
&=& \Exp\{\rx_k^o | L_{k,k}'\rx_k^o + \ldots + L_{k,1}'\rx_1^o + \rn_k,\ldots,L_{1,1}'\rx_1^o + \rn_1\}
\end{eqnarray}
since by \eqref{eq:stateControlled} and \eqref{eq:free}, the signal $\rs_i$ for $i=1,\ldots,k$ can be written as
\begin{equation}
\rs_i = L_{i,i}'\rx_i^o + \ldots + L_{1,1}'\rx_1^o + \rn_i + \underbrace{L_{i,i}'B\ru_{i-1} + \ldots + (L_{i,i}'A^{i-2}+\ldots+L_{i,2}')B\ru_1}_{\mbox{$\sigma$-$\rs_{1:i-1}$ measurable}}.
\end{equation} 
Therefore, for a given ``linear plus noise" signaling rule, the optimal transformed control input is given by $\ru_{k}^o = -K_{j,k}\Exp\{\rx_k^o| \rs_{1:k}^{}\}$.
\end{important}

In order to reverse the transformation on the control input and to provide a compact representation, we introduce
\begin{eqnarray}
\Phi_j := \begin{bmatrix} I & K_{j,\kappa}B & K_{j,\kappa}AB & \cdots & K_{j,\kappa}A^{\kappa-2}B \\ & I & K_{j,\kappa-1}B & \cdots & K_{j,\kappa-1}A^{\kappa-3}B \\ & & I & \cdots & K_{j,\kappa-2}A^{\kappa-4}B \\
& & & \ddots & \vdots \\ & & & & I \end{bmatrix},
\end{eqnarray} 
and block diagonal matrices 
\begin{equation}
K_j := \diag\{K_{j,\kappa},\ldots,K_{j,1}\} \mbox{ and } \Delta_j := \diag\{\Delta_{j,\kappa},\ldots,\Delta_{j,1}\}.
\end{equation} 
Then, \eqref{eq:change} can be written as
\begin{equation}
\sum_{k=1}^{\kappa} \Exp\|K_{j,k}^{}\rx_k^o + \ru_k^o\|_{\Delta_{j,k}}^2 = \Exp\|K_j \rx^o + \Phi_j \ru\|_{\Delta_j}^2
\end{equation}
where we have introduced the augmented vectors $\ru = \begin{bmatrix} \ru_{\kappa}' & \cdots & \ru_{1}'\end{bmatrix}'$ and $\rx^o = \begin{bmatrix} (\rx_{\kappa}^o)' & \cdots & (\rx_{1}^o)'\end{bmatrix}'$.

To recap, \playerS~and \playerR~seek to minimize, respectively, the following cost functions
\begin{eqnarray}
&U_{\playerS}&(\eta_{1:\kappa}^{\ell}, \gamma_{1:\kappa}^{}) = \Exp\|K_{\playerS}\rx^o + \Phi_{\playerS} \ru\|_{\Delta_{\playerS}}^2 + \delta_{\playerS,0},\\
&U_{\playerR}&(\eta_{1:\kappa}^{\ell}, \gamma_{1:\kappa}^{}) = \Exp\|K_{\playerR}\rx^o + \Phi_{\playerR} \ru\|_{\Delta_{\playerR}}^2 + \delta_{\playerR,0}.
\end{eqnarray}
We note the resemblance to the communication setting. Therefore following the same lines, \playerS~faces the following problem:
\begin{equation}\label{eq:SDPCont}
\min_{\eta_{1:\kappa}^{\ell} \in \Upsilon^{\ell}} \sum_{k=1}^{\kappa} \trace\{\cov\{\Exp\{\rx_k^o| \rs_{1:k}\}\}V_k\} + v_o,
\end{equation}
where $v_o := \trace\{\cov\{\rx^o(\rx^o)\}K_{\playerS}'\Delta_{\playerS}^{}K_{\playerS}^{}\} + \delta_{\playerS,0}$ and
\begin{equation}
V_k = \Xi_{k,k} + \sum_{i=k+1}^{\kappa} \Xi_{k,i}A^{i-k} + (A^{i-k})'\Xi_{i,k},
\end{equation}
where $\Xi_{k,i}\in\R^{m\times m}$ is an $m\times m$ block of $\Xi\in\R^{m\kappa \times m\kappa}$, with indexing starting from the right-bottom to the left-top, and
\begin{equation}
\Xi := M_{\playerS}'\Delta_{\playerS}M_{\playerS} - M_{\playerS}'\Delta_{\playerS}K_{\playerS} - K_{\playerS}'\Delta_{\playerS}M_{\playerS},
\end{equation}
where $M_{\playerS}^{}:=\Phi_{\playerS}^{}\Phi_{\playerR}^{-1}K_{\playerR}^{}$.

\begin{svgraybox}
The optimal linear signaling rule in control systems can be computed according to Corollary \ref{corollary:multi} based on \eqref{eq:SDPCont}.
\end{svgraybox}

\section{Uncertainty in the Uninformed Agent's Objective}
\label{sec:uncertainty}

In the deception-as-defense game $\game$, the objectives of the players are common knowledge. However, there might be scenarios where the objective of the uninformed attacker may not be known precisely by the informed defender. In this section, our goal is to extend the results in the previous sections for such scenarios with uncertainties. To this end, we consider that \playerR~has a private type $\omega\in\Omega$ governing his/her cost function and $\Omega$ is a finite set of types. For a known type of \playerR, e.g., $\omega\in\Omega$, as shown in both communication and control settings, the problem faced by the informed agent $\playerS$ can be written in an equivalent form as
\begin{equation}\label{eq:ssh}
\min_{\eta_{1:\kappa}\in\Upsilon} \sum_{k=1}^{\kappa} \trace\{H_k V_{\omega,k}\} + v_o,
\end{equation}
for certain symmetric matrices $V_{\omega,k}\in\SymMat^m$, which depend on \playerR's objective and correspondingly his/her type. If the distribution governing the type of \playerR, e.g., $\{p_{\omega}\}_{\omega\in\Omega}$, where $p_{\omega}$ denotes the probability of type $\omega\in\Omega$, were known, then the equivalence result would still hold straight-forwardly when we consider 
\begin{equation}
V_k := \sum_{\omega\in\Omega} p_{\omega} V_{\omega,k}
\end{equation} 
since \eqref{eq:ssh} is linear in $V_{\omega,k}\in\SymMat^m$. For the scenarios where the distribution governing the type of \playerR~is not known, we can defend against the worst possible distribution over the types in a robust way. In the following, we define the corresponding {\em robust} deception-as-defense game.
 
\begin{definition}
The {\bf robust deception-as-defense game} 
\begin{equation}
\game^r := (\Upsilon,\Gamma,\Omega,\{\rx_k\},\{\ry_k\},U_{\playerS}^{r},U_{\playerR}^{\omega})
\end{equation} 
is a Stackelberg game \cite{ref:Basar99} between \playerS~and \playerR, where
\begin{itemize}
\item $\Omega$ denotes the type set of \playerR,
\item $\{\rx_k\}$ denotes the information of interest,
\item $\{\ry_k\}$ denotes \playerS's (possibly noisy) measurements of the information of interest,
\item $U_{\playerS}^r$ and $U_{\playerR}^{\omega}$ are the objective functions of \playerS~and \playerR, derived based on \eqref{eq:SobjComm} and \eqref{eq:RobjComm}, or \eqref{eq:SobjCont} and \eqref{eq:RobjCont}.
\end{itemize}
In this hierarchical setting, \playerS~is the leader, who announces (and commits to) his strategies beforehand, while \playerR~stands for followers of different types, reacting to the leader's announced strategy. Players type-$\omega$ \playerR~and \playerS~select the strategies $\gamma_{1:\kappa}^{\omega}\in\Gamma$ and $\eta_{1:\kappa}\in\Upsilon$ to minimize the cost functions $U_{\playerR}^{\omega}(\eta_{1:\kappa}^{},\gamma_{1:\kappa}^{\omega})$ and
\begin{equation}
U_{\playerS}^r(\eta_{1:\kappa}^{},\{\gamma_{1:\kappa}^{\omega}\}_{\omega\in\Omega}) = \max_{p\in\Delta^{|\Omega|}} \sum_{\omega\in\Omega} p_{\omega} U_{\playerS}(\eta_{1:\kappa}^{},\gamma_{1:\kappa}^{\omega}).
\end{equation}
Type-$\omega$ \playerR~selects his/her strategy knowing \playerS's strategy $\eta_{1:\kappa}\in\Upsilon$. Let $B^{\omega}(\eta_{1:\kappa}) \subset \Gamma$ be type-$\omega$ \playerR's best reaction set to \playerS's strategy $\eta_{1:\kappa}\in\Upsilon$. Then, the strategy and best reactions pair $(\eta_{1:\kappa}^*,\{B^{\omega}(\eta_{1:\kappa}^*)\}_{\omega\in\Omega})$ attains the Stackelberg equilibrium provided that
\begin{eqnarray}
&&\eta_{1:\kappa}^* \in\argmin_{\eta_{1:\kappa}\in \Upsilon} \max\limits_{\substack{\gamma_{1:\kappa}^{\omega}\in B^{\omega}(\eta_{1:\kappa}),\\ \omega\in\Omega}} U_{\playerS}^r(\eta_{1:\kappa}^{},\{\gamma^{\omega}_{1:\kappa}\}_{\omega\in\Omega}),\label{eq:SErobust}\\
&& B^{\omega}(\eta_{1:\kappa}) \;= \argmin_{\gamma_{1:\kappa}\in \Gamma} U_{\playerR}^{\omega}(\eta_{1:\kappa}^{},\gamma_{1:\kappa}^{\omega}).
\end{eqnarray}
\end{definition}

Suppose \playerS~has access to the perfect measurement of the state. Then, in the robust deception-as-defense game $\game^r$, the equivalence result in Theorem \ref{theorem:multi} yields that the problem faced by \playerS~can be written as
\begin{equation}\label{eq:shortRobust}
\min_{S\in\Psi} \max_{p\in\Delta^{|\Omega|}} \trace\left\{S\sum_{\omega\in\Omega} p_{\omega} V_{\omega} \right\} + v_o,
\end{equation}
where we have introduced the block diagonal matrices $S:=\diag\{S_{\kappa},\ldots,S_1\}$ and $V_{\omega}:= \diag\{V_{\omega,\kappa},\ldots,V_{\omega,1}\}$, and $\Psi\subset\SymMat^{m\kappa}$ denotes the constraint set at this new high-dimensional space corresponding to the necessary and sufficient condition on the covariance of the posterior estimate. The following theorem from \cite{ref:Sayin19a} provides an algorithm to compute the optimal signaling rules within the general class of measurable policies for the communication setting, and the optimal ``linear plus noise" signaling rules for the control setting. 

\begin{theorem}
The value of the Stackelberg equilibrium \eqref{eq:SErobust}, i.e., \eqref{eq:shortRobust}, is given by $\vartheta = \min_{\omega\in\Omega} \vartheta_{\omega}$, where
\begin{equation}
\vartheta_{\omega} := \min_{S\in\Psi} \trace\{SV_{\omega}\} + v_o, \mbox{ subject to } \trace\{(V_{\omega} - V_{\omega_o})S\} \geq 0\;\forall \,\omega_o\neq \omega.
\end{equation}
Furthermore, let $\omega^*\in \argmin_{\omega\in\Omega} \vartheta_{\omega}$ and
\begin{equation}
S^* \in \argmin_{S\in\Psi} \trace\{SV_{\omega^*}\} + v_o, \mbox{ subject to } \trace\{(V_{\omega^*}-V_{\omega_o})S\} \geq 0\;\forall \, \omega_o \neq \omega^*. 
\end{equation}
Then, given $S^*\in\Psi$, we can compute the optimal signaling rule according to the equivalence result in Theorem \ref{theorem:multi}.
\end{theorem}

\begin{proof}
There exists a solution for the equivalent problem \eqref{eq:shortRobust} since the constraint sets are decoupled and compact while the objective function is continuous in the optimization arguments. Let $(S^*,p^*)$ be a solution of \eqref{eq:shortRobust}. Then, $p^*\in\Delta^{|\Omega|}$ is given by
\begin{equation}\label{eq:pStar}
p^* \in \left\{p\in\Delta^{|\Omega|} | p_{\omega} = 0 \mbox{ if } \trace\{V_{\omega}S^*\} < \max_{\omega_o\neq \omega} \trace\{V_{\omega_o}S^*\}\right\}
\end{equation} 
since the objective in \eqref{eq:shortRobust} is linear in $p\in\Delta^{|\Omega|}$. Since $p^*\in\Delta^{|\Omega|}$, i.e., a point over the simplex $\Delta^{|\Omega|}$, there exists at least one type with positive weight, e.g., $p_{\omega}^* > 0$. Then, \eqref{eq:pStar} yields 
\begin{equation}
\trace\{V_{\omega}S^*\} \geq \trace\{V_{\omega_o}S^*\},\;\forall\,\omega_o\in\Omega
\end{equation} 
and furthermore
\begin{equation}
\trace\{V_{\omega}S^*\} = \sum_{\omega_o\in\Omega}p_{\omega_o}\trace\{V_{\omega_o}S^*\},
\end{equation}
since for all $\omega_o\in\Omega$ such that $p_{\omega_o}>0$, we have $\trace\{V_{\omega_o}S^*\}=\trace\{V_{\omega}S^*\}$. Therefore, given the knowledge that in the solution $p_{\omega}^*>0$, we can write \eqref{eq:shortRobust} as
\begin{eqnarray}
\min_{S\in\Psi} \max_{p\in\Delta^{|\Omega|}} \trace\left\{S\sum_{\omega_o\in\Omega} p_{\omega_o} V_{\omega_o} \right\} + v_o =&& \min_{S\in\Psi}\trace\{V_{\omega}S\} \\
&&\mbox{s.t. } \trace\{(V_{\omega}-V_{\omega_o})S\} \geq 0\;\forall\,\omega_o\in\Omega.\nn
\end{eqnarray}
To mitigate the necessity $p_{\omega}^*>0$ in the solution of the left-hand-side, we can search over the finite set $\Omega$ since in the solution at least one type must have positive weight, which completes the proof.
\end{proof}

\begin{warning}{Irrelevant Information in Signals}
The optimization objective in \eqref{eq:shortRobust} is given by
\begin{equation}
\max_{p\in\Delta^{|\Omega|}} \trace\left\{S\sum_{\omega\in\Omega} p_{\omega} V_{\omega} \right\} + v_o,
\end{equation}
which is convex in $S\in\Psi$ since the maximum of any family of linear functions is a convex function \cite{ref:Boyd04}. Therefore, the solution $S^*\in\Psi$ may be a non-extreme point of the constraint set $\Psi$, which implies that in the optimal signaling rule \playerS~introduces independent noise. Note that Blackwell's irrelevant information theorem \cite{ref:Blackwell63,ref:Blackwell62} implies that there must also be some other (nonlinear) signaling rule within the general class of measurable policies that can attain the equilibrium without introducing any independent noise.
\end{warning}

\section{Partial or Noisy Measurements}
\label{sec:noisy}

Up to now, we have considered the scenario where \playerS~has perfect access to the underlying information of interest, but had mentioned at the beginning that results are extendable also to partial or noisy measurements, e.g.,
\begin{equation}
\ry_k = C\rx_k + \rv_k,
\end{equation}
where $C\in\R^{m\times m}$ and $\rv_k\sim\Gaussian(0,\Sigma_v)$ is Gaussian measurement noise independent of all the other parameters. In this section, we discuss these extensions, which hold under certain restrictions on \playerS's strategy space. More precisely, for ``linear plus noise" signaling rules $\eta_k^{\ell}\in\Upsilon_k^{\ell}$, $k=1,\ldots,\kappa$, the equivalence results in Theorems \ref{theorem:equivalent} and \ref{theorem:multi} hold in terms of the covariance of the posterior estimate of all the previous measurements\footnote{With some abuse of notation, we denote the vector $\begin{bmatrix} \ry_{\kappa}'  & \cdots & \ry_1' \end{bmatrix}'$ by $\ry_{1:k}\in\R^{mk}$.}, denoted by $Y_k := \cov\{\Exp\{\ry_{1:k}|\rs_{1:k}\}\}$, rather than the covariance of the posterior estimate of the underlying state $H_k = \cov\{\Exp\{\rx_k|\rs_{1:k}\}\}$. Particularly, the following lemma from \cite{ref:Sayin18d} shows that there exists a linear relation between the covariance matrices $H_k\in\SymMat^{m}$ and $Y_k\in\SymMat^{mk}$ since $\rx_k \rightarrow \ry_{1:k} \rightarrow \rs_{1:k}$ forms a Markov chain in that order.

\begin{lemma}\label{lemma:Markov}
Consider zero-mean jointly Gaussian random vectors $\rx,\ry,\rs$ that form a Markov chain, e.g., $\rx\rightarrow \ry \rightarrow \rs$ in this order. Then, the conditional expectations of $\rx$ and $\ry$ given $\rs$ satisfy the following linear relation:
\begin{equation}
\Exp\{\rx|\rs\} = \Exp\{\rx\ry'\}\Exp\{\ry\ry'\}^{\dagger}\Exp\{\ry|\rs\}.
\end{equation}
\end{lemma}

Note that $\rs_{1:k}$ is jointly Gaussian with $\rx_k$ and $\ry_{1:k}$ since $\eta_i^{\ell}\in\Upsilon_i^{\ell}$, for $i=1,\ldots,k$. Based on Lemma \ref{lemma:Markov}, the covariance matrices $H_k\in\SymMat^{m}$ and $Y_k\in\SymMat^{mk}$ satisfy
\begin{equation}
H_k^{} = D_k^{} Y_k^{} D_k',
\end{equation}
where $D_k^{} := \Exp\{\rx_k^{}\ry_{1:k}'\}\Exp\{\ry_{1:k}^{}\ry_{1:k}'\}^{\dagger} \in \R^{m\times mk}$. Furthermore, $\ry_{1:k}\in\R^{mk}$ follows the first-order auto-regressive recursion:
\begin{equation}\label{eq:Ystate}
\ry_{1:k} = \underbrace{\begin{bmatrix} \Exp\{\ry_k\ry_{1:k-1}'\}\Exp\{\ry_{1:k-1}\ry_{1:k-1}\}^{\dagger} \\ I_{m(k-1)} \end{bmatrix}}_{=:A_k^y} \ry_{1:k-1} + \begin{bmatrix} \ry_k - \Exp\{\ry_k|\ry_{1:k-1}\} \\ \vZeros_{m(k-1)} \end{bmatrix}.
\end{equation}
Therefore, the optimization problem faced by \playerS~can be viewed as belonging to the non-cooperative communication setting with perfect measurements for the Gauss-Markov process $\{\ry_{1:k}\}$ following the recursion \eqref{eq:Ystate}, and it can be written as
\begin{equation}
\min_{\eta_{1:\kappa}^{\ell}\in\Upsilon^{\ell}} \sum_{k=1}^{\kappa}\trace\{Y_kW_k\} + v_o,
\end{equation}
where $W_k^{} := D_k' V_k^{} D_k^{}$.

\begin{important}{Dimension of Signal Space}
Without loss of generality, we can suppose that the signal $\rs_k$ sent by \playerS~is $mk$ dimensional so that \playerS~can disclose $\ry_{1:k}$. To distinguish the introduced auxiliary signaling rule from the actual signaling rule $\eta_k^{\ell}$, we denote it by $\teta_k^{\ell}\in\tilde{\Upsilon}_k^{\ell}$ and the policy space $\tilde{\Upsilon}_k^{\ell}$ is defined accordingly. When the information of interest is Gaussian, for a given optimal $\teta_{1:i}^{\ell}$, we can always set the $i$th optimal signaling rule $\eta_{i}^{\ell}(\cdot)$ in the original signal space $\Upsilon_i^{\ell}$ as
\begin{equation}\label{eq:higher}
\eta_i^{\ell}(\ry_{1:i}) = \Exp\{\rx_{i}|\teta_1^{\ell}(\ry_1),\ldots,\teta_i^{\ell}(\ry_{1:i})\},
\end{equation}
almost everywhere over $\R^m$, and the right-hand-side is the conditional expectation of $\rx_i$ with respect to the random variables $\teta_1^{\ell}(\ry_1),\ldots,\teta_i^{\ell}(\ry_{1:i})$. Then, for $k=1,\ldots,\kappa$, we would obtain
\begin{equation}
\Exp\{\rx_k|\eta_1^{\ell}(\ry_1),\ldots,\eta_k^{\ell}(\ry_{1:k})\} =  \Exp\{\rx_{k}|\teta_1^{\ell}(\ry_1),\ldots,\teta_k^{\ell}(\ry_{1:k})\},
\end{equation}
almost everywhere over $\R^m$, since for $\eta_{1:\kappa}^{\ell}\in\Upsilon^{\ell}$ selected according to \eqref{eq:higher}, all the previously sent signals $\{\eta_{1}^{\ell}(\ry_1),\ldots,\eta_{k-1}^{\ell}(\ry_{1:k-1})\}$ are $\sigma$-$\{\teta_1^{\ell}(\ry_1),\ldots,\teta_{k-1}^{\ell}(\ry_{1:k-1})\}$ measurable. 
\end{important}

Based on this observation, for partial or noisy measurements, we have the equivalent problem
\begin{equation}
\min_{Y_{1:\kappa}\in\bigtimes_{k=1}^{\kappa}\SymMat^{mk}} \sum_{k=1}^{\kappa}\trace\{Y_kW_k\},\mbox{ subject to } \cov\{\ry_{1:k}\} \succeq Y_k \succeq A_k^y Y_{k-1} (A_k^y)',
\end{equation}
where $Y_0 = 0$. Given the solution $Y_{1:\kappa}^*$, we can compute the corresponding signaling rules $\teta_{1:\kappa}^{\ell}$ according to Theorem \ref{theorem:multi} and then the actual optimal signaling rule $\eta_{1:\kappa}^{\ell}\in\Upsilon^{\ell}$ can be computed by \eqref{eq:higher}.

\section{Conclusion}
\label{sec:conclusion}

In this chapter, we have introduced the deception-as-defense framework for cyber-physical systems. A rational adversary takes certain actions to carry out a malicious task based on the available information. By crafting the information available to the adversary, our goal was to control him/her to take actions inadvertently in line with the system's interest. Especially, when the malicious and benign objectives are not completely opposite of each other, as in a zero-sum game framework, we have sought to restrain the adversary to take actions, or attack the system, carrying out only the aligned part of the objectives as much as possible without meeting the goals of the misaligned part. To this end, we have adopted the solution concept of game theoretical hierarchical equilibrium for robust formulation against the possibility that advanced adversaries can learn the defense policy in the course of time once it has been widely deployed.      

We have shown that the problem faced by the defender can be written as a linear function of the covariance of the posterior estimate of the underlying state. For arbitrary distributions over the underlying state, we have formulated a necessary condition on the covariance of the posterior estimate. Then, for Gaussian state, we have shown the sufficiency of that condition since for any given symmetric matrix satisfying the necessary condition, there exists a ``linear plus noise" signaling rule yielding that covariance of the posterior estimate. Based on that, we have formulated an SDP problem over the space of symmetric matrices equivalent to the problem faced by the defender over the space of signaling rules. We have first focused on the communication setting. This equivalence result has implied the optimality of linear signaling rules within the general class of stochastic kernels. We have provided the optimal signaling rule for single stage settings analytically and provided an algorithm to compute the optimal signaling rules for dynamic settings numerically. Then, we have extended the results to control settings, where the adversary has a long-term control objective, by transforming the problem into a communication setting by restricting the space of signaling rules to linear policies plus a random term. We have also addressed the scenarios where the objective of the adversary is not known and the defender can have partial or noisy measurements of the state.

Some future directions of research include formulation of the deception-as-defense framework for
\begin{itemize}
\item robust control of systems,
\item communication or control systems with quadratic objectives over infinite horizon,
\item networked control systems, where there are multiple informed and uninformed agents,
\item scenarios where the uninformed adversary can have side-information,
\item applications in sensor selection. 
\end{itemize}    

\section*{Acknowledgement}
This research was supported by the U.S. Office of Naval Research (ONR) MURI grant N00014-16-1-2710.

\bibliographystyle{spbasic}
\bibliography{../ref}

\end{document}